\documentclass[a4paper, 11pt, thm-restate, fleqn]{article}
\usepackage{hyperref}
\usepackage{fullpage}
\usepackage{graphics}
\usepackage{mathtools}
\usepackage{braket}
\usepackage{algorithm}
\usepackage{algorithmic}
\usepackage{color}
\usepackage{amsmath}
\usepackage{amssymb}
\usepackage{amsfonts}
\usepackage{braket}
\usepackage{comment}
\usepackage{mathtools}
\usepackage{amsthm}
\usepackage{color}
\usepackage{cite}

\newtheorem{theorem}{Theorem}
\newtheorem{lemma}[theorem]{Lemma}

\newtheorem{proposition}[theorem]{Proposition}
\newtheorem{observation}[theorem]{Observation}
\theoremstyle{definition}
\newtheorem{definition}[theorem]{Definition}
\newtheorem{example}[theorem]{Example}
\newtheorem{remark}{Remark}

\title{An Approximation Algorithm for \\Maximum Stable Matching with Ties and Constraints} %TODO Please add
\author{Yu Yokoi\thanks{Principles of Informatics Research Division, National Institute of Informatics, Hitotsubashi, Chiyoda-ku, Tokyo 101-8430, Japan, E-mail: {\tt yokoi@nii.ac.jp}. }}
\date{\today}
\newcommand{\R}{\mathbf{R}}
\newcommand{\Zp}{\mathbf{Z}_{+}}

\newcommand{\w}{\partial_{W}}
\newcommand{\f}{\partial_{F}}

\newcommand{\I}{\mathcal{I}}
\newcommand{\B}{\mathcal{B}}
\newcommand{\laminar}{\mathcal{L}}

\newcommand{\M}{\mathcal{M}}
\newcommand{\C}{\mathcal{C}}

\begin{document}
\maketitle
\vspace{-4mm}
\begin{abstract}
We present a polynomial-time $\frac{3}{2}$-approximation algorithm for 
the problem of finding a maximum-cardinality stable matching in a many-to-many matching model with 
ties and laminar constraints on both sides. 
%This problem is a generalization of the NP-hard problem {\sc max-smti}, whose current best approximation ratio is $\frac{3}{2}$.
We formulate our problem using a bipartite multigraph 
whose vertices are called workers and firms, and edges are called contracts. 
%Each agent's constraints are represented by upper quotas on a laminar family on contracts.
%and they prevend us from a straightforward extension 
%of the existing $\frac{3}{2}$-approximation algorithm for {\sc max-smti}.
Our algorithm is described as the computation of a stable matching in an auxiliary instance, 
in which each contract is replaced with three of its copies and all agents have strict preferences on the copied contracts.
The construction of this auxiliary instance is symmetric for the two sides,
which facilitates a simple symmetric analysis.
We use the notion of matroid-kernel for computation in the auxiliary instance and
exploit the base-orderability of laminar matroids to show the approximation ratio.

In a special case in which each worker is assigned at most one contract 
and each firm has a strict preference, our algorithm defines a $\frac{3}{2}$-approximation mechanism 
that is strategy-proof for workers.
\end{abstract}

\section{Introduction}
The {\em college admission problem} ({\large\sc ca}) is a many-to-one generalization of the well-known 
{\em stable marriage problem} \cite{GIbook, Knuthbook, Manlovebook}, introduced by Gale and Shapley \cite{GS62}.
An instance of {\large\sc ca} involves two disjoint agent sets called students and colleges.
Each agent has a strict linear order of preference over agents on the opposite side, 
and each college has an upper quota for the number of assigned students. 
It is known that any instance of {\large\sc ca} has a stable matching, we can find it efficiently, and all stable matchings have the same cardinality.

Recently, matching problems with constraints have been studied extensively \cite{BDKW14,EHYY14,FT17,KK15,KK17}. 
Motivated by the matching scheme used in the higher education sector in Hungary, Bir{\'o} et al.~\cite{BFIM10} studied {\em {\large\sc ca} with common quotas}. 
In this problem, in addition to individual colleges, certain subsets of colleges, called {\em bounded sets}, have upper quotas. 
Such constraints are also called {\em regional caps} or {\em distributional constraints},
 and they have been studied in \cite{Goto16, KK18}. 
Meanwhile, motivated by academic hiring, Huang \cite{CCH10} introduced the {\em classified stable matching problem}. 
This is an extension of {\large\sc ca} in which each individual college has quotas for subsets of students, called {\em classes}. 
Its many-to-many generalizations have been studied in \cite{FK16, Yokoi17}.%
\footnote{In \cite{FK16, Goto16, CCH10, Yokoi17}, not only upper quotas but also lower quotas are considered. With lower quotas, the existence of stable matching is not guaranteed. In this paper, we consider only upper quotas.}
For these models, the laminar structure of constraints is commonly found to be the key to 
the existence of a stable matching. 
A family $\laminar$ of sets is called {\em laminar} if any $L,L'\in \laminar$ satisfy $L\subseteq L'$ or $L\supseteq L'$ or $L\cap L'=\emptyset$ (also called  {\em nested} or {\em hierarchical}).
In \cite{BFIM10,CCH10}, the authors showed that a stable matching exists in their models if regions or classes form laminar families, 
whereas the existence is not guaranteed in the general case. Furthermore, in the laminar case, 
a stable matching can be found efficiently, and all 
stable matchings have the same cardinality. Applications with laminar constraints have been discussed in \cite{KK18}.

\smallskip
The purpose of this paper is to introduce ties to a matching model with laminar constraints.
In the previous studies described above, the preferences of agents were assumed to be strictly ordered. 
However, ties naturally arise in real problems. Matching models with ties have been studied widely in the literature \cite{Irving94, GIbook, Manlovebook},
where the preference of an agent is said to contain a {\em tie} if she is indifferent between two or more agents on the opposite side. 
%When ties are allowed, \textcolor{red}{there are several definitions of stability, and we adopt the one called weak stability as the definition of the stability with the existence of ties.} 
When ties are allowed, the existence of a stable matching is maintained; however, stable matchings vary in cardinalities. 
As it is desirable to produce a large matching in practical applications, we consider the problem of finding a maximum-cardinality stable matching.

Such a problem is known to be difficult even in the simple matching model without constraints.
The problem of finding a maximum stable matching in the setting of {\em stable marriage with ties and incomplete lists}, called {\sc max-smti},  
is NP-hard \cite{IMMM99,MIIMM02}, as is obtaining an approximation ratio 
within $\frac{33}{29}$ \cite{Yanagisawa07}. 
For its approximability, several algorithms with improved approximation 
ratios have been proposed \cite{IMY07,IMY08,Kiraly11,Kiraly13, Mcdermid09,Paluch14}. 
The current best ratio is $\frac{3}{2}$ by a polynomial-time algorithm proposed by McDermid \cite{Mcdermid09} 
as well as linear-time algorithms proposed by Paluch \cite{Paluch14} and Kir\'aly \cite{Kiraly13}. 
%Moreover, the integrality gap of a natural IP formulation for {\sc max-smti} is shown to be at least $\frac{3}{2}-o(1)$ in \cite{IMY14}, 
%which rules out the possibility of further improvement by techniques such as rounding and primal-dual algorithms.
The $\frac{3}{2}$-approximability extends to the settings of {\large\sc ca} with ties \cite{Kiraly13}
and the student-project allocation problem with ties \cite{CM18}.

\bigskip
\noindent\textbf{\textsf{Our Contribution.}}~
We present a polynomial-time $\frac{3}{2}$-approximation algorithm for 
the problem of finding a maximum-cardinality stable matching in a many-to-many matching model with 
ties and laminar constraints on both sides. 
We call this problem {\sc max-smti-lc} and formulate it using a bipartite multigraph, 
where we call the two vertex sets {\em workers} and {\em firms}, respectively, 
and each edge a {\em contract}. Each agent has upper quotas on a laminar family defined on incident contracts.
Our formulation can deal with each agent's constraints, such as {\em classified stable matching}.
Furthermore, distributional constraints such as {\em {\large\sc ca} with common quotas} can be handled by considering a 
dummy agent that represents a consortium of the agents on one side (see Remark~\ref{rem:example} at the end of Section~\ref{sec:formulation}).
Our algorithm runs in $O(k\cdot|E|^2)$ time, where $E$ is the set of contracts and $k$ is the maximum level of nesting of laminar constraints. 
The {\em level of nesting} of a laminar family $\laminar$ is the maximum length of a chain 
$L_1\subsetneq L_2\subsetneq \cdots\subsetneq L_k$ of members of $\laminar$;
hence, $k\leq|E|$.
%\medskip
%\noindent\textbf{\textsf{Algorithm.}}~

Our algorithm is described as the computation of a stable matching in an auxiliary instance.
Here, we explain the ideas underlying the construction of the auxiliary instance,
which is inspired by the algorithms of Kir\'aly~\cite{Kiraly13} and Hamada, Miyazaki, and Yanagisawa~\cite{HMY19}.

First, we briefly explain Kir\'aly's $\frac{3}{2}$-approximation algorithm for {\sc max-smti} \cite{Kiraly13}.
In this algorithm, each worker makes proposals from top to bottom in her list sequentially, 
as with the worker-oriented Gale--Shapley algorithm.
A worker rejected by all firms is given a second chance for proposals.
Each firm prioritizes a worker in the second cycle over a worker in the first cycle if 
they are tied in its preference list. 
This idea of {\em promotion} is used to handle ties in firms' preference lists.
To handle ties in workers' lists, Kir\'aly's algorithm lets
each worker prioritize a currently unmatched firm over a currently matched firm if they are tied in her preference list. 
This priority rule depends on the states of firms at each moment,
which makes the algorithm complicated when we introduce constraints on both sides.

Then, we introduce the idea of the algorithm of Hamada et al. \cite{HMY19}, 
who proposed a worker-strategy-proof algorithm for {\sc max-smti} that 
attains the $\frac{3}{2}$-approximation ratio
when ties appear only in workers' lists.
They modified Kir\'aly's algorithm such that 
each worker's proposal order is predetermined and is not affected by the history of the algorithm.
Their algorithm can be seen as a Gale--Shapley-type algorithm in which each worker makes proposals twice to each firm in a tie before proceeding to the next tie, 
and each firm prioritizes second proposals over first proposals regardless of its preference.
By combining their algorithm with the promotion operation of Kir\'aly's algorithm,
we obtain a Gale--Shapley-type algorithm in which each worker makes at most three proposals to each firm.

%If we consider a generalization of this algorithm to {\sc max-smti-lc}, 
%then constraints on the worker side prevent top-to-bottom sequential proposals,
%Moreover, the method to handling ties in the worker side becomes 
%complicated as it changes priority in dynamically depending on firms' constraints and current states.
%%Thus, the order of proposals depend on the constraints and the history of the algorithm.
%To avoid analyzing such a complicated algorithm, 
%we reduce our problem to a computation of a stable matching in an auxiliary instance
%and provide a static analysis.

Based on these observations, we propose
a method for transforming a {\sc max-smti-lc} instance $I$ into an auxiliary instance $I^*$, which is also a {\sc max-smti-lc} instance.
Each contract $e_i$ in $I$ is replaced with three copies $x_i,y_i,z_i$ in $I^*$.
Each agent has a strict preference on the copied contracts, which reflects the priority rules in the algorithms of Kir\'aly and  Hamada et al.
%Constraints in $I^*$ is defined so that a set of contracts is  and can take at most one from $x_i,y_i,z_i$
The instance $I^*$ has an upper bound $1$ for each triple $\{x_i,y_i,z_i\}$ and 
also has constraints corresponding to those in $I$. 
The construction of $I^*$ is completely symmetric for workers and firms.
%while the algorithms of Kir\'aly~\cite{Kiraly13} and Hamada et al. \cite{HMY19} are asymmetric.
We show that, for any stable matching $M^*$ of  $I^*$,
its projection $M\coloneqq \set{e_i| \{x_i,y_i,z_i\}\cap M^*\neq \emptyset}$ is 
a $\frac{3}{2}$-approximate solution for $I$.
Both the stability and the approximation ratio of $M$ are implied by the stability of $M^*$ in $I^*$,
and the process of computing $M^*$ is irrelevant.
Thus, our method enables us to conduct a symmetric and static analysis even with constraints.

Because the auxiliary instance $I^*$ has no ties, we can find a stable matching of $I^*$ efficiently 
by using the matroid framework of Fleiner~\cite{Fleiner01,Fleiner03}.
In the analysis of the approximation ratio, we exploit  the fact that the family of feasible sets defined by 
laminar constraints forms a matroid with a property called {\em base-orderability}.

In the last section, we show that the result of Hamada et al. \cite{HMY19} mentioned above 
is generalized to a many-to-one matching setting with laminar constraints on the firm side.
In other words, if we restrict {\sc max-smti-lc} such that each worker is assigned at most one contract 
and each firm has a strict preference, then we can provide 
a worker-strategy-proof mechanism that returns a $\frac{3}{2}$-approximate solution.
We obtain this conclusion using the strategy-proofness result of Hatfield and Milgrom \cite{HM05}.

\bigskip
\noindent\textbf{\textsf{Paper Organization.}}~
The remainder of this paper is organized as follows.
Section~\ref{sec:formulation}  formulates our matching model, while
Section~\ref{sec:alg} describes our algorithm.
Section~\ref{sec:matroids} presents a lemma on base-orderable matroids that is 
the key to our proof of the approximation ratio. 
Sections~\ref{sec:correctness} and \ref{sec:complexity} are 
devoted to the proofs of correctness and time complexity, respectively.
Section~\ref{sec:SP} investigates strategy-proof approximation mechanisms for our model.

\smallskip
Throughout the paper, we denote the set of non-negative integers by $\Zp$.
For a subset $S\subseteq E$ and an element $e\in E$, we denote
$S+e\coloneqq S\cup\{e\}$ and $S-e\coloneqq S\setminus\{e\}$.

\section{Problem Formulation}\label{sec:formulation}
\label{sec:Problem}
An instance of the {\em stable matching with ties and laminar constraints}, which we call {\sc smti-lc}, 
is a tuple $I=(W,F,E,\{\laminar_a,q_a,P_a\}_{a\in W\cup F})$ defined as follows.
Let $W$ and $F$ be disjoint finite sets called {\em workers} and {\em firms}, respectively.
We call $a\in W\cup F$ an {\em agent} when we do not distinguish between workers and firms. 
We are provided a set $E$ of {\em contracts}. Each contract $e\in E$ is associated with one worker and one firm, 
denoted by $\w(e)$ and $\f(e)$, respectively. 
Multiple contracts are allowed to exist between a worker--firm pair.
Then, $(W,F;E)$ is represented as a bipartite multigraph in which $W$ and $F$ are vertex sets, and 
each $e\in E$ is an edge connecting $\w(e)$ and $\f(e)$.
For each $a\in W\cup F$, we denote the set of associated contracts by $E_a$, i.e., 
\[E_w\coloneqq\set{e\in E|\w(e)=w}~~(w\in W),\qquad E_f\coloneqq\set{e\in E|\f(e)=f}~~(f\in F).\]
Then, the family $\set{E_w|w\in W}$ forms a partition of $E$, as does $\set{E_f|f\in F}$.

Each agent $a\in W\cup F$ has a laminar family $\laminar_a$ of subsets of $E_a$ 
and a quota function $q_a\colon\laminar_a\to \Zp$.
%We call each $L\in \laminar_a$ a {\em bounded set} and $q_a(L)$ the {\em quota} for $L$.
For any subset $M\subseteq E$ of contracts and an agent $a\in W\cup F$, we denote by $M_{a}\coloneqq M\cap E_a$
the set of contracts assigned to $a$. 
We say that $M$ is {\em feasible} for $a\in W\cup F$ if 
\[\forall L\in \laminar_a: |M_{a}\cap L|\leq q_a(L).\]
A set $M\subseteq E$ is called a {\em matching} if it is feasible for all agents in $W\cup F$.
%We simply say that $M$ is {\em feasible} if it is feasible for all agents $a\in W\cup F$.

Each agent $a\in W\cup F$ has a preference list $P_a$ that consists of all elements in $E_a$ and may contain ties.
In this paper, a preference list is written in one row, from left to right according to preference,
where two or more contracts with equal preference are included in the same parentheses. 
For example, if %$E_a=\{e_1,e_2,e_3,e_4\}$ for an agent $a\in W\cup F$ and her 
the preference list $P_a$ of an agent $a\in W\cup F$ is represented as
\begin{align*}
&P_a:\  e_2 \ \ ( \ e_1\ \ e_4 \ )\  \ e_3,
\end{align*}
then $e_2$ is $a$'s top choice, $e_1$ and $e_4$ are the second choices with equal preference, and $e_3$ is the last choice.
For contracts $e, e'\in E_a$, we write $e \succ_{a} e'$ if $a$ prefers $e$ to $e'$.
Furthermore, we write $e \succeq_{a} e'$ if $e \succ_{a} e'$ or $a$ is indifferent between $e$ and $e'$ (including the case $e=e'$).

For a matching $M\subseteq E$, a contract $e\in E\setminus M$, and an associated agent $a\in\{\w(e),\f(e)\}$, 
we say that $e$ is {\em free for $a$} in $M$ if 
\smallskip
\begin{itemize}
\item $M_{a}+e$ is feasible for $a$, or 
\item there is $e'\in M_{a}$ such that $e\succ_a e'$ and $M_{a}+e-e'$ is feasible for $a$.
\end{itemize}
\smallskip
In other words, a contract $e$ is free for an agent $a$ if $a$ has an incentive to add $e$ to the current assignment possibly at the expense of some less preferred contract $e'$.
A contract $e\in E\setminus M$ {\em blocks} $M$ if $e$ is free for both $\w(e)$ and $\f(e)$. 
A matching $M$ is {\em stable} if there is no contract in $E\setminus M$ that blocks $M$.
%\footnote{This corresponds to {\em weak stability} among several stability notions defined for the settings with ties.}

The goal of our problem {\sc max-smti-lc} is
to find a maximum-cardinality stable matching for a given {\sc smti-lc} instance.
Because {\sc max-smti-lc} is a generalization of the NP-hard problem {\sc max-smti}, 
we consider the approximability.
%For {\sc max-smti}, it is known that any stable matching is a $2$-approximate solution \cite{MIIMM02}
%and that one can find a stable matching by breaking ties arbitrarily and 
%applying the Gale-Shapley algorithm for the resultant instance \cite{GIbook, Manlovebook}.
%These facts generalize to {\sc max-smti-lc} as shown in Proposition~\ref{prop:2-approx} in Section~\ref{sec:Analysis}.
%Therefore, arbitrary tie-breaking method is sufficient to achive $2$-approximation.
Similarly to the case of {\sc max-smti}, for {\sc max-smti-lc}, a $2$-approximate solution can be easily 
obtained using an arbitrary tie-breaking method (see Proposition~\ref{prop:2-approx} in the Appendix).
In the next section, we present a $\frac{3}{2}$-approximation algorithm.

\bigskip
\begin{remark}\label{rem:example}
We demonstrate that {\large\sc smti-lc} includes several models investigated in previous works, 
which implies that our algorithm finds $\frac{3}{2}$-approximate solutions 
for the problems of finding maximum-cardinality stable matchings in those models with ties.

First, {\large\sc smti} and the stable $b$-matching problem are special cases 
such that $E\subseteq W\times F$ and $\laminar_a=\{E_a\}$ for every $a\in W\cup F$.
Furthermore, the two-sided laminar classified stable matching problem \cite{FK16,CCH10}, if lower quotas are absent, 
is a special case with $E\subseteq W\times F$.

To represent {\large\sc ca} with laminar common quotas \cite{BFIM10},
let $W$ be the set of students and let $F\coloneqq\{f\}$, where $f$ is regarded as a 
consortium of all colleges in  $C$.
The set of contracts is defined by  
$E\coloneqq\set{(w,f,c)| \text{a college $c\in C$ is acceptable for a student $w\in W$}}$,
where $\w\!~\!(e)\!=\!w$, $\f(e)\!=\!f$ for any $e=(w,f,c)$. Note that $E=E_f$.
A quota for a region $C'\subseteq C$ is then represented as
a quota for the set $\set{(w,f,c)\in E| c\in C'}\subseteq E_f$.
Thus, laminar common quotas can be represented as constraints on a laminar family on $E_f$.
%Our formulation can also deal with constraints on students side.

For the student-project allocation problem \cite{CM18},
let $W$ and $F$ be the sets of students and lecturers, respectively,
and $E\coloneqq\set{(w,f,p)|\text{a project $p$ acceptable for $w\in W$ is offered by $f\in F$}}$.
Let $E_{f,p}\subseteq E_f$ be the set of contracts associated with a project $p$ offered by a lecturer $f$.
Then, the lecturer's upper quota and projects' upper quotas define 
two-level laminar constraints on
the family $\laminar_f=\{E_f\}\cup\set{E_{f,p}|\text{$p$ is offered by $f$}}$.

For the above-mentioned settings, we can appropriately set the preferences of 
agents such that the stability in the previous works
coincides with the stability in {\large\sc smti-lc}.
\end{remark}

\section{Algorithm}\label{sec:alg}
Our approximation algorithm for {\sc max-smti-lc} consists of three steps:
(i) construction of an auxiliary instance,
(ii) computation of any stable matching of this auxiliary instance,
and (iii) mapping the obtained matching to a matching of the original instance.
In what follows, we describe how to construct an auxiliary instance $I^*$ from a given instance $I$ 
and how to map a matching of $I^*$ to that of $I$.

Let $I=(W,F,E,\{\laminar_a,q_a,P_a\}_{a\in W\cup F})$ be an instance of {\sc max-smti-lc},
where the set $E$ of contracts is represented as $E=\set{e_i|i=1,2,\dots,n}$. 
We construct an auxiliary instance $I^*=(W,F,E^*,\{\laminar^*_a,q^*_a,P^*_a\}_{a\in W\cup F})$, which is also an {\sc smti-lc} instance; however, each preference list $P^*_a$ does not contain ties.

The sets of workers and firms in $I^*$ are the same as those in  $I$.
The set $E^*$ of contracts in $I^*$ is given as $E^*=\set{x_i,y_i,z_i|i=1,2,\dots,n}$, 
where $x_i$, $y_i$, and $z_i$
are copies of $e_i$; hence, $\w(x_i)=\w(y_i)=\w(z_i)=\w(e_i)$ and $\f(x_i)=\f(y_i)=\f(z_i)=\f(e_i)$.
We define a mapping $\pi:2^{E^*}\to 2^E$ by 
$\pi(S^*)=\set{e_i|\{x_i,y_i,z_i\}\cap S^*\neq \emptyset}$ for any $S^*\subseteq E^*$.

For any agent $a\in W\cup F$, the laminar family $\laminar^*_a$ and the quota function $q_a^*:\laminar^*_a\to \Zp$ 
are defined as follows.
For each $e_i\in E_a$, we have $\{x_i,y_i,z_i\}\in \laminar^*_a$ and $q^*_a(\{x_i,y_i,z_i\})=1$.
For each $L\in \laminar_a$, we have $L^*:=\set{x_i,y_i,z_i|e_i\in L}\in \laminar^*_a$
and $q^*_a(L^*)=q_a(L)$. These are all that $\laminar^*_a$ contains.
Then, for any set $M^*\subseteq E^*$ of contracts, we see that $M^*$ is feasible for $a$ in $I^*$ if and only if
$M^*$ contains at most one copy of each $e_i\in E_a$ and 
the set $\pi(M^*)$ is feasible for $a$ in $I$. 

The preference list $P^*_w$ of each worker $w\in W$ is defined as follows.
Take a tie $(e_{i_1} e_{i_2} \cdots e_{i_\ell})$ in $P_w$. 
We replace it with a strict linear order of $2\ell$ contracts 
$x_{i_1} x_{i_2} \cdots x_{i_\ell} y_{i_1} y_{i_2} \cdots y_{i_\ell}$.
Apply this operation to all the ties in $P_w$,
where we regard a contract not included in any tie as a tie of length one.
Next, at the end of the resultant list, 
append the original list $P_w$ with each $e_i$ replaced with $z_i$ and 
all the parentheses omitted.
Here is a demonstration. If the preference list of a worker $w$ is 
\begin{align*}
&P_w:\  ( \ e_2\ \ e_6 \ ) \ \ e_1\ \ (\ e_3 \ \  e_4\ ),
\end{align*}
then  her list in $I^*$ is 
\begin{align*}
&P^*_w:\ \ x_2\ x_6 \ y_2\ y_6 \ x_1\ y_1 \ x_3 \ x_4 \ y_3\ y_4\ z_2\ z_6\ z_1\ z_3\ z_4.
\end{align*}
The preference list $P^*_f$ of each firm $f\in F$ is defined in the same manner, 
where the roles of $x_i$ and $z_i$ are interchanged.
For example, if the preference list of a firm $f$ is
\begin{align*}
&P_f:\  \ e_3\ \ (\ e_2 \ \ e_4\ \ e_7 \ )\ \ e_5,
\end{align*}
then its list in $I^*$ is 
\begin{align*}
&P^*_f:\ \ z_3\ y_3 \ z_2\ z_4 \ z_7\ y_2 \ y_4 \ y_7 \ z_5\ y_5\ x_3\ x_2\ x_4\ x_7\ x_5.
\end{align*}

Thus, we have defined the auxiliary instance $I^*$.
As this is again an {\sc smti-lc} instance, 
a stable matching of $I^*$ is defined as before.
The existence of a stable matching of $I^*$ is guaranteed by the existing framework 
of Fleiner \cite{Fleiner01, Fleiner03}, as will be explained in Section~\ref{sec:complexity}.
Here is the main theorem of this paper, which states that
any stable matching of $I^*$ defines a $\frac{3}{2}$-approximate solution for $I$.

\begin{theorem}\label{thm:main}
For a stable matching $M^*$ of $I^*$, let $M\coloneqq\pi(M^*)$. 
Then, $M$ is a \mbox{stable~matching} of $I$ with $|M|\geq \frac{2}{3}|M_{\rm OPT}|$, where $M_{\rm OPT}$
is a maximum-cardinality stable matching of $I$.
\end{theorem}
We prove Theorem~\ref{thm:main} in Section~\ref{sec:correctness}.
This theorem guarantees the correctness of Algorithm~\ref{alg1}.
% which constructs $I^*$ and compute its stable matching.

\begin{algorithm}[htb]\caption{\sf ~$\frac{3}{2}$-approximation algorithm for {\sc max-smti-lc}} \label{alg1}
\begin{algorithmic}[1]
\REQUIRE An instance $I=(W,F,E,\{\laminar_a,q_a,P_a\}_{a\in W\cup F})$.
\ENSURE A stable matching $M$ with $|M|\geq \frac{2}{3}|M_{\rm OPT}|$, where $M_{\rm OPT}$ is an optimal solution.
\STATE Construct an auxiliary instance $I^*$.
\STATE Find any stable matching $M^*$ of $I^*$.
\STATE Let $M=\pi(M^*)$ and return $M$.
\end{algorithmic}
\end{algorithm}

Clearly, the first and third steps of Algorithm~\ref{alg1} can be performed efficiently.
Furthermore, the second step can be executed in polynomial time by 
applying the generalized Gale--Shapley algorithm of Fleiner \cite{Fleiner01, Fleiner03}. 
In Section~\ref{sec:complexity}, we will explain this more precisely and 
present the time complexity represented in the following theorem.

\begin{theorem}\label{thm:alg}
One can find a stable matching $M$ of $I$ with $|M|\geq \frac{2}{3}|M_{\rm OPT}|$
 in $O(k\cdot|E|^2)$ time, where $M_{\rm OPT}$ is a maximum-cardinality stable matching
and $k$ is the maximum level of nesting of laminar families $\laminar_a~(a\in W\cup F)$.
\end{theorem}

%In Section~\ref{sec:correctness} we show the approximation ratio   
%by using the fact that the laminar constraints induce matroids with a property 
%called {\em base-orderability}. 

\section{Base-orderable Matroids}\label{sec:matroids}
For the proofs of Theorems~\ref{thm:main} and ~\ref{thm:alg},
we introduce some concepts related to matroids
(see, e.g., Oxley \cite{Oxley} for more information on matroids). 

For a finite set $E$ and a family $\I\subseteq 2^{E}$, a pair $(E,\I)$ is called a {\em matroid} if the following three conditions hold:
(I1) $\emptyset\in \I$, (I2) $S\subseteq T\in \I$ implies $S\in \I$, and (I3) for any $S,T\in \I$ with $|S|<|T|$,
there exists $e\in T\setminus S$ such that $S+e\in \I$.

For a matroid $(E,\I)$, each member of $\I$ is called an {\em independent set}.
An independent set is called a {\em base} if it is inclusion-wise maximal in $\I$. 
We denote the family of all bases by $\B$.
By the matroid axiom (I3), it follows that $|B_1|=|B_2|$ holds for any bases $B_1, B_2\in \B$.
%A matroid is called {\em base-orderable} \cite{Brualdi69, BS68} if it satisfies the following property.
%\footnote{There is also a notion of {\em strong base-orderability} \cite{BS16}. 
%In this paper, base-orderability always refers to weak base-orderability.}
\begin{definition}[Base-orderable Matroid]
A matroid $(E,\I)$ is called {\em base-orderable} if for any two bases $B_1, B_2\in \B$,
there exists a bijection $\varphi\colon B_1\to B_2$ with the property that,
for every $e\in B_1$, both $B_1-e+\varphi(e)$ and $B_2+e-\varphi(e)$ are bases.
\end{definition}
A class of base-orderable matroids includes {\em gammoids} (see \cite{Brualdi71} and \cite[Theorem 42.12]{Schr03}),
and gammoids include laminar matroids described below (see \cite{FO17} and \cite[Section 2.3.1]{Finkelstein11}).
\begin{example}[Laminar Matroid]\label{prop:laminar-matroid}
For a laminar family $\laminar$ on $E$ 
and a function $q\colon\laminar\to \Zp$, define
%be a family of subsets satisfying all capacity constraints, i.e.,
$\I=\set{S\subseteq E|\forall L\in \laminar: |S\cap L|\leq q(L)}$.
Then,  $(E,\I)$ is a base-orderable matroid.
\end{example}
A matroid is {\em laminar} if it can be defined in the above-mentioned manner 
for some $\laminar$ and $q$. 
%In the setting of {\sc max-smti-lc}, 
%if we define $\I_a=\set{S\subseteq E_a|\forall L\in \laminar_a: |S\cap L|\leq q_a(L)}$ for any agent $a\in W\cup F$, 
%then $(E_a, \I_a)$ is a laminar matroid, and a set $M\subseteq E$ is feasible for $a$ if and only if $M_a\in \I_a$.

Base-orderability is known to be closed under the following operations
(see, e.g., \cite{BS16, Ingleton77}).
\smallskip
\begin{description}
\setlength{\itemsep}{1.5mm}
\item[Contraction.]\hspace{-2mm}\footnote{\mbox{Contraction is defined for any subset of $E$ \cite{Oxley}; 
however this paper uses only contraction by independent sets.}}
~For a matroid $(E, \I)$ and any $S\in \I$, define %a family $\I_{S}$ on $E\setminus S$ by 
$\I_{S}\coloneqq\set{T\subseteq E\setminus S|S\cup T\in \I}$.
Then, $(E\setminus S, \I_{S})$ is a matroid.
If $(E, \I)$ is base-orderable, then so is $(E\setminus S,\I_{S})$.

\item[Truncation.] 
~For a matroid $(E, \I)$ and any integer $p\in \Zp$, define %a family $\I_{S}$ on $E\setminus S$ by 
$\I_{p}\coloneqq\set{S\in \I||S|\leq p}$.
Then, $(E, \I_{p})$ is a matroid.
If $(E, \I)$ is base-orderable, then so is $(E,\I_{p})$.

\item[Direct Sum.]	 
~For matroids $(E_j, \I_j) ~(j=1,2,\dots,\ell)$
such that $E_j$ are all pairwise disjoint,\\
let $E\coloneqq E_1\cup E_2\cup\cdots \cup E_\ell$
and $\I\coloneqq\set{S_1\cup S_2\cup \cdots\cup S_\ell| S_j\in \I_j~(j=1,2,\dots,\ell)}$.
Then, $(E, \I)$ is a matroid.
% with circuit family $\C=\C_1\cup \C_2\cup\cdots\cup \C_k$, where $\C_j$ is that of $(E_j, \I_j)$.
If all $(E_j, \I_j)$ are base-orderable, then so is $(E,\I)$.
\end{description}

%For two matroids $(E, \I_1)$ and $(E, \I_2)$, 
%their intersection $(E, \I_1\cap \I_2)$ is not a matroid in general.
On the intersection of two base-orderable matroids, we show the following property, 
which plays a key role in proving the $\frac{3}{2}$-approximation ratio of our algorithm.
This generalizes the fact that, if (one-to-one) bipartite matchings $M$ and $N$ satisfy $|M|<\frac{2}{3}|N|$, 
then $M\triangle N$ contains a connected component that forms an alternating path of length at most three.
\begin{lemma}\label{lem:triple}
For base-orderable matroids $(E, \I_1)$ and $(E, \I_2)$, suppose that $S,T\in \I_1\cap \I_2$ and 
$|S|<\frac{2}{3}|T|$. If $S+e\not\in \I_1\cap \I_2$ for every $e\in T\setminus S$, then there exist distinct elements $e_{i}, e_{j}, e_{k}$ such that
$e_{i}, e_{k}\in T\setminus S$,
$e_{j}\in S\setminus T$, and the following conditions hold:
\begin{itemize}
\smallskip
\setlength{\itemsep}{0mm}
\item $S+e_{i}\in \I_1$,
\item both $S+e_{i}-e_{j}$ and $T-e_{i}+e_{j}$ belong to $\I_2$,
\item both $S-e_{j}+e_{k}$ and $T+e_{j}-e_{k}$ belong to $\I_1$,
\item $S+e_{k}\in \I_2$.
\end{itemize}
\end{lemma}
\begin{proof}
By the matroid axiom (I3), there is a subset $A_1\subseteq T\setminus S$
such that $|A_1|=|T|-|S|$ and $S_1\coloneqq S\cup A_1\in \I_1$.
Then, $|S_1|=|T|$; hence, $|S_1\setminus T|=|T\setminus S_1|$.
Let $(E',\I'_1)$ be a matroid obtained from $(E, \I_1)$ by 
contracting $S_1\cap T$ and truncating with size $|S_1\setminus T|$,
i.e., $E'=E\setminus(S_1\cap T)$ and $\I'_1\coloneqq\set{R\subseteq E'| R\cup (S_1\cap T)\in \I_1, |R|\leq |S_1\setminus T|}$.
Then, $S_1\setminus T$ and $T\setminus S_1$ are bases of $(E',\I'_1)$.
As $(E',\I'_1)$ is base-orderable, 
there is a bijection $\varphi_1\colon S_1\setminus T\to T\setminus S_1$ such that 
both $(S_1\setminus T)-e+\varphi_1(e)$ and $(T\setminus S_1)+e-\varphi_1(e)$ are bases of $(E',\I'_1)$
for every $e\in S_1\setminus T$.
By the definition of $\I'_1$, this implies
that both $S-e+\varphi_1(e)$ and $T+e-\varphi_1(e)$ belong to $\I_1$ for every $e\in S_1\setminus T$.
By the same argument, there exists $A_2\subseteq T\setminus S$  
such that $|A_2|=|T|-|S|$ and $S_2\coloneqq S\cup A_2\in \I_2$, and 
there exists a bijection $\varphi_2\colon S_2\setminus T\to T\setminus S_2$ such that
both $S-e+\varphi_2(e)$ and $T+e-\varphi_2(e)$ belong to $\I_2$
for every $e\in S_2\setminus T$. 

We represent $\varphi_1$ and $\varphi_2$ using a bipartite graph as follows. 
Note that, for each $\ell\in\{1,2\}$, we have $S_\ell\setminus T=S\setminus T$ and 
$T\setminus S_\ell=T\setminus (S\cup A_\ell)\subseteq T\setminus S$.
Let $S\setminus T$ and $T\setminus S$ be two vertex sets
and let \mbox{$M_\ell\coloneqq\set{(e,\varphi_\ell(e))|e\in S\setminus T}$} for $\ell=1,2$.
Then, each $M_\ell$ is a \mbox{one-to-one} matching that covers $S\setminus T$ and $T\setminus (S\cup A_\ell)$.
Note that the sets $A_1, A_2\subseteq S\setminus T$ are mutually disjoint since, otherwise, 
some $e\in A_1\cap A_2$ satisfies $S+e\in \I_1\cap\I_2$, which contradicts the assumption.
Then,  $|T\setminus (S\cup A_1\cup A_2)|=|T\setminus S|-|A_1|-|A_2|=|T\setminus S|-2|T|+2|S|$.
Therefore, at most $2(|T\setminus S|-2|T|+2|S|)$ vertices in $S\setminus T$ 
are adjacent to $T\setminus (S\cup A_1\cup A_2)$ via the edges in $M_1\cup M_2$.
Because $|S\setminus T|-2(|T\setminus S|-2|T|+2|S|)=-3|S|+2|T|+|S\cap T|>-3\cdot\frac{2}{3}|T|+2|T|+|S\cap T|\geq0$,
there exists $\tilde{e}\in S\setminus T$ that is not adjacent to $T\setminus (S\cup A_1\cup A_2)$
via $M_1\cup M_2$.
This implies that  $\varphi_2(\tilde{e})\in A_1$ and $\varphi_1(\tilde{e})\in A_2$;
hence, $S+\varphi_2(\tilde{e})\in \I_1$ and $S+\varphi_1(\tilde{e})\in \I_2$.
Let $e_{i}\coloneqq\varphi_2(\tilde{e})$, $e_{j}\coloneqq \tilde{e}$, and $e_{k}\coloneqq\varphi_1(\tilde{e})$.
Then, these three elements satisfy all the required conditions.
\end{proof}
%From the proof, we can observe that $\frac{2}{3}$ is the maximum number for this statement to hold.

\section{Correctness}\label{sec:correctness}
This section is devoted to showing Theorem~\ref{thm:main},
which establishes the correctness of Algorithm~\ref{alg1}.

As in Section~\ref{sec:alg}, let $I$ be an {\sc smti-lc} instance with $E=\set{e_i|i=1,2,\dots,n}$ and 
let $I^*$ be the auxiliary instance $I^*$, 
whose contract set is $E^*=\set{x_i,y_i,z_i|i=1,2,\dots,n}$.

For any agent $a\in W\cup F$, let $E^*_a=\set{x_i,y_i,z_i|e_i\in E_a}$ and 
define families $\I_a$  and  $\I^*_a$  by
\begin{align*}
\I_a&=\set{~\!S~\subseteq E_a|\forall L~\!\in \laminar_a: ~\!|~\!S~\!\cap L~\!|\leq q_a(~\!\!L~\!)},\\
\I^*_a&=\set{S^*\subseteq E^*_a|\forall L^*\!\in \laminar^*_a: |~\!S^*\!\cap L^*|\leq q^*_a(L^*)},
\end{align*}
i.e., $\I_a$ and $\I^*_a$ are the families of feasible sets in $I$ and $I^*$, respectively.
Then, $(E_a, \I_a)$ and $(E^*_a,\I^*_a)$ are laminar matroids and base-orderable.
The definitions of $\laminar^*_a$ and $q^*_a$ imply the following fact.
Recall that $\pi:2^{E^*}\to 2^E$ is defined by $\pi(S^*)=\set{e_i|\{x_i,y_i,z_i\}\cap S^*\neq \emptyset}$. 
%for any $S^*\subseteq E^*$.
%By the definitions of $\laminar^*_a$ and $q^*_a$, we can observe the following fact.
\begin{observation}\label{obs:feasible}
For a set $S^*\subseteq E^*_a$, we have $S^*\in \I^*_a$
if and only if $|\{x_i,y_i,z_i\}\cap S^*|\leq 1$ for every $e_i\in E_a$ and $\pi(S^*)\in \I_a$.
\end{observation}

Take any stable matching $M^*$ of $I^*$ and let $M\coloneqq \pi(M^*)$.
%We show that $M$ is a stable matching of $I$ with $|M|\geq \frac{2}{3}|M_{\rm OPT}|$, where $M_{\rm OPT}$ is a maximum-cardinality stable matching of $I$.
As $M^*$ is feasible in $I^*$, it contains at most one copy of each contract $e_i$.
For any $e_i\in M$, we denote by $\pi^{-1}(e_i)$ the unique element in $\{x_i,y_i,z_i\}\cap M^*$.

By the definitions of the preference lists $\{P^*_a\}_{a\in W\cup F}$ in $I^*$, we can observe the following properties.
For any agent $a\in W\cup F$ and contracts $e,e'\in E^*_a$, we write $e\succ^*_a e'$ 
if $a$ prefers $e$ to $e'$ with respect to $P^*_a$. 
Recall that $P^*_a$ does not contain ties,  while $P_a$ may contain.
\begin{observation}\label{obs:preference}
For any $e_i\in E\setminus M$ and $e_j\in M$, the following conditions hold.  
\begin{itemize}
\smallskip
\setlength{\itemsep}{2mm}
\item For any agent $a\in W\cup F$, if $e_i,e_j\in E_a$ and $e_i\succ_a e_j$, 
then $y_i\succ_a^* \pi^{-1}(e_j)$ holds regardless of which of $\{x_i,y_i,z_i\}$ is $\pi^{-1}(e_i)$.
\item For any worker $w\in W$, if $e_i,e_j\in E_w$ and $\pi^{-1}(e_j)\succ_w^* x_i$,
then we have either \\
{\rm[} $\pi^{-1}(e_j)=x_j$ and $e_j\succeq_w e_i$ {\rm]}  or {\rm[} $\pi^{-1}(e_j)=y_j$ and $e_j\succ_w e_i$ {\rm]}.
\item For any firm $f\in F$, if $e_i,e_j\in E_f$ and $\pi^{-1}(e_j)\succ_f^* z_i$,
then we have either\\
{\rm[}~$\pi^{-1}(e_j)=z_j$ and $e_j\succeq_f e_i$ {\rm]} or {\rm[} $\pi^{-1}(e_j)=y_j$ and $e_j\succ_f e_i$ {\rm]}.
\end{itemize}
\end{observation}

First, we show the stability of $M$ in $I$.
For each agent $a\in W\cup F$, we write $M^*_a= M^*\cap E^*_a$, which implies that $\pi(M^*_a)=M_a$.

\begin{lemma}\label{lem:stable}
The set $M$ is a stable matching of $I$.
\end{lemma}
\begin{proof}
Since $M^*$ is feasible for all agents in $I^*$, 
Observation~\ref{obs:feasible} implies that $M=\pi(M^*)$ is feasible for all agents in $I$, i.e., $M$ is a matching in $I$.

Suppose, to the contrary, that $M$ is not stable. 
Then,  some contract $e_i\in E\setminus M$ blocks $M$. 
Let $w=\w(e_i)$ and $f=\f(e_i)$.
Then, $e_i$ is free for both $w$ and $f$ in $M$.
We now show that $y_i$ is free for both $w$ and $f$ in $M^*$, which contradicts the stability of $M^*$.

As $e_i$ is free for $w$ in $I$, we have (i) $M_w+e_i\in \I_w$ or 
(ii) there exists $e_j\in M_a$ such that $e_i\succ_w e_j$ and $M_a+e_i-e_j\in \I_w$.
Note that $e_i \in E\setminus M$ implies $\{x_i,y_i,z_i\}\cap M^*=\emptyset$.
In case (i), we have $\pi(M^*_w+y_i)=M_w+e_i\in \I_w$, which implies $M^*_w+y_i\in \I^*_w$
; hence, $y_i$ is free for $w$ in $M^*$.
In case (ii), we have $\pi(M^*_w+y_i-\pi^{-1}(e_j))=M_w+e_i-e_j\in \I_w$,
which implies $M^*_w+y_i-\pi^{-1}(e_j)\in \I^*_w$. 
Furthermore, as $e_i\succ_w e_j$, the first statement of Observation~\ref{obs:preference} implies $y_i\succ^*_w \pi^{-1}(e_j)$. 
Thus, in each case, $y_i$ is free for $w$ in $M^*$.

Similarly, we can show that $y_i$ is free for $f$ in $M^*$.
Thus, $y_i$ blocks $M^*$, a contradiction.
\end{proof}

Next, we show the approximation ratio using Lemma~\ref{lem:triple}.
%For this purpose, we define two base-orderable matroids on $E$.
Note that $\set{E_w|w\in W}$ is a partition of $E$, as is $\set{E_f|f\in F}$.
Let $(E,\I_{W})$ be the direct sum of base-orderable matroids 
$\set{(E_w,\I_w)|w\in W}$ and 
$(E,\I_{F})$ be the direct sum of $\set{(E_f,\I_f)|f\in F}$.
Then, they are both base-orderable matroids on $E$.

By the definitions of $\I_{W}$ and $\I_{F}$, for any subset $N\subseteq E$, we have $N\in \I_{W}\cap \I_{F}$ if and only if 
$N_a\coloneqq N\cap E_a$ is feasible for each $a\in W\cup F$, i.e., $N$ is a matching.
Furthermore, for any matching $N\in \I_{W}\cup \I_{F}$ and 
contract $e_i\in E\setminus N$, which is associated with a worker $w=\w(e_i)$ (and a firm $f=\f(e_i)$),
the condition $N+e_i\in \I_W$ is equivalent to $N_w+e_i\in \I_w$.
In addition, if $N+e_i\not\in \I_W$, we have $N+e_i-e_j\in \I_W$ if and only if 
$e_j\in N_w$ and $N_w+e_i-e_j\in \I_w$.
The same statements hold when $w$ and $W$ are replaced with $f$ and $F$, respectively.

\begin{lemma}\label{lem:approx}
The set $M$ satisfies $|M|\geq \frac{2}{3}|M_{\rm OPT}|$, where $M_{\rm OPT}$
is a maximum-cardinality stable matching of $I$.
\end{lemma}
\begin{proof}
Set $N\coloneqq M_{\rm OPT}$ for notational simplicity. 
Since $M$ and $N$ are stable matchings, $M,N\in \I_{W}\cap \I_{F}$.
In addition, $M+e_i\not\in \I_{W}\cap \I_{F}$ for any $e_i\in N\setminus M$
since, otherwise, $e_i$ blocks $M$.
Suppose, to the contrary, that $|M|<\frac{2}{3}|N|$.
Then, by Lemma~\ref{lem:triple} and the definitions of $\I_{W}$ and $\I_{F}$, 
there exist three contracts $e_i,e_j,e_k$ such that
$e_i,e_k\in N\setminus M$,
$e_j\in M\setminus N$, and the following conditions hold:
\begin{itemize}
\smallskip
\setlength{\itemsep}{0mm}
\item $M_{w}+e_i\in \I_w$,
\item both $M_{f}+e_i-e_j$ and $N_f-e_i+e_j$ belong to $\I_{f}$, 
\item both $M_{w'}-e_j+e_k$ and $N_{w'}+e_j-e_k$ belong to $\I_{w'}$, 
\item $M_{f'}+e_k\in \I_{f'}$, 
\end{itemize}
\smallskip
where $w=\w(e_i)$, $f=\f(e_i)=\f(e_j)$, $w'=\w(e_j)=\w(e_k)$, $f'=\f(e_k)$.

\smallskip
Since $e_i\not\in M$ and $M_{w}+e_i\in \I_w$, we have $M^*_w+z_i\in \I^*_w$; hence, $z_i$ is free for the worker $w=\w(z_i)$ in $M^*$.
Then, the stability of $M^*$ implies that $z_i$ is not free for the firm $f=\f(z_i)$.
Since $\pi(M^*_f+z_i-\pi^{-1}(e_j))=M_{f}+e_i-e_j\in \I_{f}$ implies $M^*_{f}+z_i-\pi^{-1}(e_j)\in \I^*_{f}$,
we should have $\pi^{-1}(e_j)\succ^*_f z_i$.
Then, the third statement of Observation~\ref{obs:preference} implies that
we have either [$\pi^{-1}(e_j)=z_j$ and $e_j\succeq_f e_i$] or [$\pi^{-1}(e_j)=y_j$ and $e_j\succ_f e_i$].

Meanwhile, since $e_k\not\in M$ and $M_{f'}+e_k\in \I_{f'}$,
we have $M^*_{f'}+x_k\in \I^*_{f'}$; hence, $x_k$ is free for the firm $f'=\w(x_k)$ in $M^*$.
As $M^*$ is stable, then $x_k$ is not free for the worker $w'=\w(x_k)$.
Since $\pi(M^*_{w'}+x_k-\pi^{-1}(e_j))=M_{w'}+e_k-e_j\in \I_{w'}$ 
implies $M^*_{w'}+x_k-\pi^{-1}(e_j)\in \I^*_{w'}$,
we should have $\pi^{-1}(e_j)\succ^*_{w'} x_k$.
Then, the second statement of Observation~\ref{obs:preference} implies that
we have either [$\pi^{-1}(e_j)=x_j$ and $e_j\succeq_{w'} e_k$] or [$\pi^{-1}(e_j)=y_j$ and $e_j\succ_{w'} e_k$].

Because we cannot have $\pi^{-1}(e_j)=z_j$ and $\pi^{-1}(e_j)=x_j$ simultaneously,
we must have $\pi^{-1}(e_j)=y_j$, $e_j\succ_f e_i$, and $e_j\succ_{w'} e_k$.
As we have $N_f-e_i+e_j\in \I_{f}$ and $N_{w'}+e_j-e_k\in \I_{w'}$, 
these preference relations imply that $e_j$ blocks $N$,
which contradicts the stability of $N$.
\end{proof}

\begin{proof}[Proof of Theorem~\ref{thm:main}]
Combining Lemmas~\ref{lem:stable} and \ref{lem:approx}, we obtain Theorem~\ref{thm:main}.
\end{proof}

\section{Time Complexity}\label{sec:complexity}
We explain how to implement the second step of Algorithm~\ref{alg1}
and estimate its time complexity, which establishes Theorem~\ref{thm:alg}.
For this purpose, we introduce the notion of a matroid-kernel, which is a matroid generalization of 
a stable matching proposed by Fleiner \cite{Fleiner01, Fleiner03}.
Note that it is defined not only for base-orderable matroids but for general matroids.

\subsection{Matroid-kernels}\label{sec:matroid-kernel}
A triple $\M=(E,\I, \succ)$ is called an {\em ordered matroid} if $(E, \I)$ is a matroid and 
$\succ$ is a strict linear order on $E$.
For an ordered matroid $\M=(E,\I, \succ)$ and an independent set $S\in \I$, 
an element $e\in E\setminus S$ is said to be {\em dominated} by $S$ in $\M$ if 
$S+e\not\in \I$ and there is no element $e'\in S$ such that $e\succ e'$ and $S+e-e'\in \I$.

Let $\M_1=(E,\I_1, \succ_1)$ and $\M_2=(E, \I_2, \succ_2)$ be two ordered matroids on the same ground set $E$.
Then, a set $S\subseteq E$ is called an {\em $\M_1 \M_2$-kernel} if $S\in \I_1\cap \I_2$ and 
any element $e\in E\setminus S$ is dominated by $S$ in $\M_1$ or $\M_2$.

In \cite{Fleiner01}, an algorithm for finding a matroid-kernel has been described using choice functions defined as follows.
For an ordered matroid $\M=(E,\I, \succ)$, give indices of elements in $E$ such that 
$E=\{e^1,e^2,\dots,e^n\}$ and $e^1\succ e^2\succ \cdots \succ e^n$.
Define a function $C_{\M}:2^E\to 2^E$ by letting $C_{\M}$ be the output of 
the following greedy algorithm for every $S\subseteq E$. 
Let $T^{0}:=\emptyset$ and define $T^{\ell}$ for $\ell=1,2,\dots,n$ by 
\begin{equation*}
T^{\ell}\coloneqq
 \begin{cases}
T^{\ell-1}+e^{\ell}& \text{if~~} e^{\ell}\in S \text{~~and~~} T^{\ell-1}+e^{\ell}\in \I,\\
T^{\ell-1}& \text{otherwise};
\label{eq:greedy2}
\end{cases}
\end{equation*}
then, let $\C_{\M}(S)\coloneqq T^{n}$.

Let $C_{\M_1}$, $C_{\M_2}$ be the choice functions defined from
$\M_1=(E,\I_1, \succ_1)$, $\M_2=(E, \I_2, \succ_2)$, respectively.
In \cite[Theorem~2]{Fleiner01}, Fleiner showed that an $\M_1\M_2$-kernel can be found using the following algorithm,
which can be regarded as a generalization of the Gale--Shapley algorithm.
First, set $R\gets \emptyset$. Then, repeat the following three steps:
(1) $S\gets C_{\M_1}(E\setminus R)$,
(2) $T\gets C_{\M_2}(S\cup R)$, and
(3) $R\gets (S\cup R)\setminus T$.
Stop the repetition if $R$ is not changed at (3) and return $T$ at that moment.
In terms of the ordinary Gale--Shapley  algorithm, $R$, $S$, and $T$ correspond to
the sets of contracts that are rejected by firms thus far, proposed by workers, and accepted by firms, respectively.

\begin{theorem}[Fleiner \cite{Fleiner01, Fleiner03}]\label{thm:Fleiner}
For any pair of ordered matroids $\M_1$ and $\M_2$ on the same ground set $E$,
there exists an $\M_1\M_2$-kernel. One can find an $\M_1\M_2$-kernel in $O(|E|\cdot {\rm EO})$ time,
where ${\rm EO}$ is the time required to compute $C_{\M_1}(S)$ and $C_{\M_2}(S)$ for any $S\subseteq E$.  
\end{theorem}

\subsection{Implementation of Our Algorithm}\label{sec:reduction}
We show that the second step of Algorithm~\ref{alg1} is reduced to a computation of a matroid-kernel.

For an auxiliary instance $I^*$ defined in Section~\ref{sec:Problem},
note that $\set{E^*_w|w\in W}$ is a partition of $E^*$ and let $(E^*, \I^*_W)$ be 
the direct sum of $\{(E^*_w,\I^*_w)\}_{w\in W}$.
Furthermore, let $\succ_W$ be a strict linear order on $E^*$ that is 
consistent with the workers' preferences $\{P^*_w\}_{w\in W}$ in $I^*$.
For example, obtain $\succ_{W}$ by concatenating the lists $P^*_w$ of all workers in an arbitrary order.
%That is, if $\w(e_i)=\w(e_j)=w$ for some $w\in W$, 
%then $e_i\succ_w e_j$ if and only if $e_i\succ_W e_j$. 
%If $\w(e)\neq\w(e')$, we do not care about the ordering between $e$ and $e'$.
Then, $\M_W=(E^*, \I^*_W, \succ_W)$ is an ordered matroid on the contract set $E^*$.
As $\set{E^*_f|f\in F}$ is also a partition of $E^*$,
we can define an ordered matroid 
$\M_F=(E^*,\I^*_F, \succ_F)$ in the same manner from $\{(E^*_f,\I^*_f)\}_{f\in F}$ and 
$\{P^*_f\}_{f\in F}$.

We show that $\M_W\M_F$-kernels are equivalent to stable matchings of $I$.
This has already been shown in several previous works \cite{FK16, Yokoi17}. 
We present  a proof for the completeness.

\begin{lemma}\label{lem:equivalence}
$M^*\subseteq E^*$ is a stable matching of $I^*$ if and only if $M^*$ is an $\M_W\M_F$-kernel.
\end{lemma}
\begin{proof}
By the definitions of $(E^*, \I^*_W)$ and $(E^*, \I^*_f)$, 
a set $M^*\subseteq E^*$ is feasible for all agents in $I^*$ if and only if $M^*\in \I^*_W\cap \I^*_F$.
Recall that a contract $e\in E^*\setminus M^*$ is free for the associated worker $w\coloneqq \w(e)$ 
if $M^*_w+e\in \I^*_w$ or there exists $e'\in M^*_w$ such that 
$e\succ^*_w e'$ and $M^*_w+e-e'\in \I^*_w$.
By the definition of $\I^*_W$,
we have $M^*_w+e\in \I^*_w$ if and only if $M^*+e\in \I^*_W$.
In addition, if $M^*_w+e\not\in \I^*_w$, then $M_w^*+e-e'\in \I^*_w$ holds for $e'\in M_w$ 
if and only if $M^*+e-e'\in \I^*_W$.
Because $\succ_W$ is consistent with $\succ^*_w$, 
these imply that $e$ is free for $w=\w(e)$ in $M^*$ if and only if $e$ is not dominated by $M^*$ in $\M_W$.
Similarly, we can show that $e$ is free for the associated firm $f\coloneqq \f(e)$ in $M^*$
if and only if $e$ is not dominated by $M^*$ in $\M_F$.
Thus, the equivalence holds.
\end{proof}

\begin{lemma}\label{lem:oracle}
For any subset $S^*\subseteq E^*$, we can compute $C_{\M_W}(S^*)$ and $C_{\M_F}(S^*)$ in $O(k^*\cdot |E^*|)$ time,
where $k^*$ is the maximum level of nesting of laminar families $\laminar^*_a~(a\in W\cup F)$.  
\end{lemma}
\begin{proof}
We only explain the computation of $C_{\M_W}(S^*)$ because that of $C_{\M_F}(S^*)$ is similar.

Let $\laminar$ be the union of $\{\laminar^*_w\}_{w\in W}$
and define $q:\laminar\to \Zp$ by setting $q(L)=q^*_w(L)$ for each $w\in W$ and $L\in \laminar^*_w$.
Then, $\laminar$ is a laminar family on $E^*$ and the matroid $(E^*, \I^*_W)$ is defined by $\laminar$ and $q$. 
The maximum level of nesting of $\laminar$ is again $k^*$. 

Referring to \cite{BFIM10}, we represent $\laminar$ by 
a forest $G$ whose node set is $\set{v_{L}|L\in \laminar}$.
Node $v_L$ is the parent of $v_{L'}$ in $G$
if $L\subseteq L'$ and there is no $L''\in \laminar$ such that $L\subsetneq L''\subsetneq L'$.
Note that $\laminar$ contains the set $\{x_i,y_i,z_i\}$ for every $e_i\in E$, 
which is inclusion-wise minimal in $\laminar$.
Therefore, the node $v_i\coloneqq v_{\{x_i,y_i,z_i\}}$ is a leaf for any $e_i\in E$, and any leaf has this form. 

We compute the sequence $T^{0}, T^{1},\dots, T^{|E^*|}$ of sets in the definition of $C_{\M_W}(S^*)$ as follows. 
For each $v_L$, we store a pointer to its parent, 
the value of $q(L)$, and 
the value of $|T^{\ell-1}\cap L|$.
For each $e^{\ell}\in E^*$, 
we have $T^{\ell-1}+e^{\ell}\in \I^*_W$ if and only if there is no ancestor node $v_{L}$ of 
$v_i$ with $q(L)=|T^{\ell-1}\cap L|$, where $v_i$ is the leaf with $e^{\ell}\in \{x_i,y_i,z_i\}$.
Then, we can check whether $T^{\ell-1}+e^{\ell}\in \I^*_W$ in $O(k^*)$ time by following the path of 
the parent pointers from $v_i$.
When $T^{\ell}=T^{\ell-1}+e^{\ell}$, 
we update the stored values $|T^{\ell-1}\cap L|$ to $|T^{\ell}\cap L|$ for each $L\in \laminar$ with $e^{\ell}\in L$.
This is also performed in $O(k^*)$ time by following the path of the parent pointers.
\end{proof}

\begin{proof}[Proof of Theorem~\ref{thm:alg}]
As we have Theorem~\ref{thm:main}, what is left is to show the time complexity.
The set $E^*$ of contracts in $I^*$ satisfies $|E^*|=3|E|$.
The maximum level of nesting of laminar families $\laminar^*_a$ in $I^*$ is $k+1$.
By Theorem~\ref{thm:Fleiner} and Lemmas~\ref{lem:equivalence} and \ref{lem:oracle},
then the second step of Algorithm~\ref{alg1} is computed in $O((k+1)\cdot |E^*|^2)=O(k\cdot|E|^2)$ time.
Since the first and third steps can be performed in $O(k\cdot|E|^2)$ time, Algorithm~\ref{alg1} runs in $O(k\cdot|E|^2)$ time.
\end{proof}

\begin{remark}
Our analysis depends on the fact that the feasible set family defined by laminar constraints 
forms the independent set family of a base-orderable matroid. 
Actually, we can extend Theorem~\ref{thm:main} to a setting where
the family of feasible sets of each agent $a\in W\cup F$ is represented by
the independent set family $\I_a$ of an arbitrary base-orderable matroid.
To construct $I^*$ in this case, we define $E^*$ and $\{P^*_a\}_{a\in W\cup F}$ as in Section~\ref{sec:alg}
and define the feasible set family $\I^*_a$ by $\I^*_a=\set{S^*\subseteq E_a^*| |\{x_i,y_i,z_i\}\cap S^*|\leq 1~\text{ for any }e_i\in E_a \text{ and }\pi(S^*)\in \I_a}$.
We can easily show that $(E_a^*, \I^*_a)$ is also a base-orderable matroid and
apply the arguments in Sections~\ref{sec:correctness} and \ref{sec:complexity}, except Lemma~\ref{lem:oracle}.
Given a membership oracle for each $\I_a$ available, 
Algorithm~\ref{alg1} runs in $O(\tau\cdot |E|^2)$ time in this case, where $\tau$ is the time for an oracle call.
%It is open whether Algorithm~\ref{alg1} attains the $\frac{3}{2}$-approximability for general matroids.
\end{remark}

\section{Strategy-Proof Approximation Mechanisms}\label{sec:SP}
In this section, we investigate approximation ratios for {\sc max-smti-lc} attained by strategy-proof mechanisms.
First, note that our setting {\sc smti-lc} is a generalization of the stable marriage model of Gale and Shapley \cite{GS62}; hence, Roth’s impossibility theorem \cite{Roth86} implies that there is no mechanism that returns
a stable matching and is strategy-proof for agents on both sides.
As with many existing works on strategy-proofness in two-sided matching models, 
we consider one-sided strategy-proofness in the setting of many-to-one matching.
Many-to-one matching models have various applications such as assignment of residents to hospitals \cite{Roth84b, RP99}
and students to high schools \cite{APR05, APRS05, APR09}.
In such applications, strategy-proofness for residents or students is a desirable property preventing their strategic behavior.

\subsection{Model and Definitions}
We define a setting of {\sc smti-olc}, which is a many-to-one variant of {\sc smti-lc}.
(Here, {\sc olc} stands for ``one-sided laminar constraints'').
In {\sc smti-olc}, each worker is assigned at most one contract and hence has no laminar constraints.
An instance of {\sc smti-olc}  
is described as $I=(W,F,E,\{P_w\}_{w\in W}, \{\laminar_f,q_f,P_f\}_{f\in F})$.
To consider strategies of workers, we slightly change the assumption on each $P_w$.
In Section~\ref{sec:Problem}, it is assumed that $P_w$ contains all contracts in $E_w$.
Here, we allow each worker to submit a preference list $P_w$ that is defined on any subset of $E_w$
and regard contracts not appearing in $P_w$ as unacceptable for $w$.
Let $E^{\circ}$ be the set of acceptable contracts,
that is, $E^{\circ}=\set{e\in E|\text{$e$ appears in $P_w$, where $w=\w(e)$}}$.

A set $M\subseteq E$ is called a {\em matching} if $M\subseteq E^{\circ}$, $|M_w|\leq 1$ for every worker $w\in W$, 
and $M$ is feasible for every firm $f\in F$.
For a matching $M$, a contract $e\in E\setminus M$ {\em blocks} $M$ if it is free for both $\w(e)$ and $\f(e)$,
where we say that $e$ is {\em free} for the associated worker $w\coloneqq \w(e)$ if $e\in E^{\circ}$ and either
$w$ is assigned no contract in $M$ or prefers $e$ to the contract assigned in $M$.
A matching $M$ is {\em stable} if there is no contract that blocks $M$.
The auxiliary instance $I^*=(W,F,E^*,\{P^*_w\}_{w\in W}, \{\laminar^*_f,q^*_f,P^*_f\}_{f\in F})$ of $I$
is defined similarly as in Section~\ref{sec:alg}.

We remark that {\sc smti-olc} is indeed as a special case of {\sc smti-lc}, 
although the assumption on workers' preference lists is slightly different from that of {\sc smti-lc}.
From an {\sc smti-olc} instance $I$,
define $I^{\circ}=(W,F,E^{\circ},\{\laminar^{\circ}_a,q^{\circ}_a,P^{\circ}_a\}_{a\in W\cup F})$ as follows. 
For each worker $w\in W$, set $\laminar^{\circ}_w=\{E^{\circ}_w\}$, $q^{\circ}_w(E^{\circ}_w)=1$, and $P^{\circ}_w=P_w$.
For each firm $f\in F$, set $\laminar^{\circ}_f=\set{L\cap E^{\circ}|L\in \laminar_f}$, 
$q^{\circ}_f(L\cap E^{\circ})=q_f(L)$ for each $L\in \laminar_f$, 
and let $P^{\circ}_f$ be the restriction of $P_f$ on $E^{\circ}_f$ 
(i.e., delete the elements in $E_f\setminus E^{\circ}_f$ from $P_f$).
Then, $I^{\circ}$ is an instance of {\sc smti-lc} in Section~\ref{sec:formulation}.
By definition, we can see that a subset $M\subseteq E$ is a stable matching of $I$ 
if and only if it is a stable matching of $I^{\circ}$.
Therefore, we can apply Algorithm~\ref{alg1} to {\sc smti-olc} instances.

For subsets $M,N\subseteq E$, a worker $w\in W$, and a preference list $P_w$,
we say that $w$ {\em weakly prefers} $M$ to $N$ with respect to $P_w$ if
either (i) $w$ is assigned a contract appearing in $P_w$ only in $M$ or 
(ii) $w$ is assigned a contract appearing in $P_w$ in 
both $M$ and $N$ and does not strictly prefer the one assigned in $N$ with respect to $P_w$.
A stable matching $M$ of an {\sc smti-olc} instance $I$ is {\em worker-optimal} 
if, for any other stable matching $N$ of $I$, every worker $w$ weakly prefers $M$ to $N$.

A {\em mechanism} is a mapping from {\sc smti-olc} instances to matchings.
Here, we define the {\em worker-strategy-proofness} of a mechanism.
%A mechanism is called {\em strategy-proof for workers} if it gives workers no incentive to
%misrepresent their preferences. The precise definition follows. 
Let $A$ be a mechanism. 
For any instance $I$ and any worker $w$,
let $I'$ be an instance obtained from $I$ by replacing $w$'s list $P_w$ with some other list $P'_w$. 
Let $M$ and $M'$ be the outputs of $A$ for instances $I$ and $I'$, respectively.
%Note that each worker is assigned at most one contract by the assumption (a). 
We say that $A$ is {\em worker-strategy-proof} if $w$ weakly prefers $M$ to $M'$ with respect to the original list $P_w$ 
regardless of the choices of $I$, $w$, and $P'_{w}$.

\subsection{Approximation Mechanisms}
Before providing our results on {\sc smti-olc},
we introduce some existing results on special cases.

We first present a result on the setting without ties.
As shown in Section~\ref{sec:reduction}, for an {\sc smti-olc} instance in which all agents have strict preferences, 
stable matchings can be represented as matroid-kernels. 
Therefore, the existing results on matroid-kernel \cite{Fleiner01, Fleiner03} imply that all the stable matchings have the same cardinality and there is a unique worker-optimal stable matching.
%the following for the case without ties: (1) all the stable matchings of a given instance have the same cardinality and (2) any instance admints a unique worker-optimal stable matching.
The following lemma is a a natural consequence of the results in \cite{KTY18}.
\begin{lemma}\label{lem:SP}
In a restriction of {\sc smti-olc} in which all agents have strict preferences, 
a mechanism that returns the worker-optimal stable matching is worker-strategy-proof.
\end{lemma}
For the completeness, Appendix~\ref{app:HM} provides the proof of Lemma~\ref{lem:SP},   
which uses the fact that {\sc smti-olc} can be reduced to the model of Hatfield and Milgrom \cite{HM05} if there are no ties.

Next, we introduce the results of Hamada et al.~\cite{HMY19} on {\sc max-smti},
which is a special case of {\sc max-smti-olc} in which every agent is assigned at most one contract.
%Workers and firms here correspond to men and women in their terminologies~\cite{HMY19}.
\begin{theorem}[Hamada et al.~\protect{\cite[Theorem 2]{HMY19}}]\label{thm:HMY1}
For {\sc max-smti}, there is a worker-strategy-proof mechanism that returns a $2$-approximate solution.
On the other hand, for any $\epsilon>0$, there is no worker-strategy-proof mechanism that returns a $(2-\epsilon)$-approximate solution. 
\end{theorem}
\begin{theorem}[Hamada et al.~\protect{\cite[Theorem 4]{HMY19}}]\label{thm:HMY2}
For a restriction of {\sc max-smti} in which ties appear in only workers' preference lists, 
there is a worker-strategy-proof mechanism that returns a $\frac{3}{2}$-approximate solution.
On the other hand, for any $\epsilon>0$, there is no worker-strategy-proof mechanism that returns a $(\frac{3}{2}-\epsilon)$-approximate solution. 
\end{theorem}

The first statement of Theorem~\ref{thm:HMY1} is attained by a naive mechanism that first breaks ties 
in an increasing order of the indices and then finds the worker-optimal stable matching of the resultant instance.
This method naturally extends to the setting of {\sc smti-olc} and yields the following theorem. 
See Appendix\ref{app:2-approx} for the proof.
\begin{theorem}\label{thm:2-approx}
For {\sc smti-olc}, 
there is a worker-strategy-proof mechanism that 
returns a stable matching $M$ with $|M|\geq \frac{1}{2}|M_{\rm OPT}|$
in $O(k\cdot|E|^2)$ time, where $M_{\rm OPT}$ is a maximum-cardinality stable matching 
and $k$ is the maximum level of nesting of $\laminar_f~(f\in F)$.
\end{theorem}

Since {\sc smti-olc} is a generalization of {\sc smti},
the second statement (i.e., the hardness part) of Theorem~\ref{thm:HMY1} immediately extends to {\sc max-smti-olc}.
%Furthermore, it is also shown in \cite[Corollary 3]{HMY19} that the second statement of Theorem~\ref{thm:HMY1} holds even if ties appear in only firms' preference lists.
Therefore, for the general {\sc smti-olc}, there is no worker-strategy-proof mechanism with an approximation ratio better than $2$.

However, in a special case in which firms' lists contain no ties,
Algorithm~\ref{alg1} in Section~\ref{sec:alg} defines a worker-strategy-proof mechanism whose approximation ratio is $\frac{3}{2}$.
That is, we can extend the first statement of Theorem~\ref{thm:HMY2} to the setting of {\sc smti-olc}.
According to the second statement of Theorem~\ref{thm:HMY2}, this is the best approximation ratio attained by a
worker-strategy-proof mechanism.
\begin{theorem}\label{thm:SP}
For a restriction of {\sc smti-olc} in which ties appear in only workers' lists, 
there is a worker-strategy-proof mechanism that 
returns a stable matching $M$ with $|M|\geq \frac{2}{3}|M_{\rm OPT}|$
in $O(k\cdot|E|^2)$ time, where $M_{\rm OPT}$ is a maximum-cardinality stable matching 
and $k$ is the maximum level of nesting of laminar families $\laminar_f~(f\in F)$.
\end{theorem}

We provide a mechanism that meets the requirements in Theorem~\ref{thm:SP}.
Our mechanism is regarded as a possible realization of Algorithm~\ref{alg1}.
In the second step of Algorithm~\ref{alg1}, we should choose the worker-optimal stable matching of 
the auxiliary instance $I^*$.
Our mechanism is described as follows.
\medskip
\begin{enumerate}
\item Given an instance $I$ (in which ties appear in only workers' lists), construct $I^*$.
\item Find the worker-optimal stable matching $M^*$ of $I^*$.
\item Let $M=\pi(M^*)$ and return $M$.
\end{enumerate}
\medskip
In the proof of Theorem~\ref{thm:Fleiner} (Fleiner \cite[p.113]{Fleiner01}), it is shown that one can find
the $\M_1$-optimal $\M_1\M_2$-kernel in $O(|E|\cdot {\rm EO})$ time.
The arguments in Section~\ref{sec:complexity} then imply that one can find the worker-optimal stable matching of $I^*$ in $O(k\cdot|E|^2)$ time. 
As we have Theorem~\ref{thm:alg}, showing the strategy-proofness of the above-mentioned 
mechanism completes the proof of Theorem~\ref{thm:SP}.
To this end, we show the following lemma. 
\begin{lemma}\label{lem:xy}
Let $I$ be an {\sc smti-olc} instance with $E=\set{e_i|i=1,2,\dots,n}$ and let 
$I^*$ be the auxiliary instance.
If ties appear in only workers' lists in $I$,
then the worker-optimal stable matching $M^*$ of  $I^*$ satisfies 
$M^*\cap \set{z_i|i=1,2,\dots,n}=\emptyset$.
\end{lemma}
\begin{proof}
Suppose, to the contrary, that $z_i\in M^*$ for some index $i$.
Then $N\coloneqq M^*-z_i+y_i$ is a matching of $I^*$ and 
$w\coloneqq \w(z_i)=\w(y_i)$ prefers $N$ to $M^*$.
We intend to show that $N$ is stable in $I^*$.
Take any $e\in E^*\setminus N=(E^*\setminus M^*)+z_i-y_i$.
If $e=z_i$, then it does not block $N$ because $y_i\succ^*_w z_i$.
If $e\neq z_i$, then the assignment of $\w(e)$ does not change in $M^*$ and $N$,
and hence $e$ can block $N$ only if $f\coloneqq \f(e)=\f(z_i)$ and $z_i \succ^*_f e \succ^*_f y_i$.
This is impossible because no contract lies between $z_i$ and $y_i$ in $P^*_f$ as the list $P_f$ 
of the firm $f$ is strict.
Thus, $N$ is a stable matching of $I^*$,
which contradicts the worker-optimality of $M^*$.
\end{proof}

\begin{proof}[Proof of Theorem~\ref{thm:SP}]
As we have Theorem~\ref{thm:alg}, what is left is to show that our mechanism is worker-strategy-proof.
Let $I=(W,F,E,\{P_w\}_{w\in W}, \{\laminar_f,q_f,P_f\}_{f\in F})$ be an instance of the setting in the statement 
and let $E=\set{e_i|i=1,2,\dots,n}$.
Furthermore, let $I'$ be obtained from $I$ by replacing $P_w$ with some other list $P'_w$.
Let $M^*$ and $N^*$ be the worker-optimal stable matchings of the auxiliary instances defined from $I$ and $I'$, respectively.
Note that the two auxiliary instances have no ties and they differ only in the preference list of $w$. 
Then, Lemma~\ref{lem:SP} implies that $w$ weakly prefers $M^*$ to $N^*$ with respect to $P^*_w$.
In other words, either (i) $w$ is assigned a contract on $P^*_w$ only in $M^*$, or 
(ii) $w$ is assigned a contract on $P^*_w$ in both $M^*$ and $N^*$ and does not strictly 
prefer the one assigned in $N^*$ w.r.t. $P^*_w$.
By Lemma~\ref{lem:xy}, $w$ is not assigned a contract of type $z_i$ in $M^*$ or $N^*$.
Then, the definition of $P^*_w$ implies that $w$ weakly prefers $\pi(M^*)$ to $\pi(N^*)$ w.r.t. $P_w$.
Thus the mechanism is worker-strategy-proof.
\end{proof}

\section*{Acknowledgments}
The author thanks the anonymous reviewers for their helpful comments.
The author was supported by JSPS KAKENHI Grant Number JP18K18004. 
This work was partially supported by the joint project of Kyoto University and Toyota Motor Corporation, titled ``Advanced Mathematical Science for Mobility Society''.

%\bibliographystyle{C:/Users/Yu/Dropbox/myplain}%           
%\bibliography{C:/Users/Yu/Dropbox/myrefs}%
%\bibliographystyle{C:/Users/yu451/Dropbox/myplain}%
\bibliography{myrefs_ISAAC2021}%

%\clearpage
\appendix
\section{Omitted Proofs}
\subsection{Proof of Lemma~\ref{lem:SP}}\label{app:HM}
We prove Lemma~\ref{lem:SP} in Section~\ref{sec:SP}, which states that, 
if the preference lists of all agents are strict in {\sc smti-olc}, 
then a mechanism that always returns the worker-optimal stable matching is strategy-proof for workers.
This is a natural consequence of the results shown in previous works \cite{Goto16, KTY18}.
We provide a proof for the completeness.

For this purpose, we introduce the model of Hatfield and Milgrom~\cite{HM05},
which we call {\em the HM model}, using our notations and terminologies.
An instance of the HM model is given by $(W,F,E,\{P_w\}_{w\in W}, \{C_f\}_{f\in F})$.
The difference from {\sc smti-olc} is that 
$P_w$ should be strict and each firm has a choice function $C_f:2^{E_f}\to 2^{E_f}$
instead of the triple $\{\laminar_f,q_f,P_f\}$. 
A function $C_f:2^{E_f}\to 2^{E_f}$ is called a {\em choice function} if $C_f(S)\subseteq S$ for any $S\subseteq E_f$. 

%Let $E^{\circ}$ be the set of contracts appearing in $\{P_w\}_{w\in W}$.
A stable matching in the HM model is defined similarly to that in {\sc smti-olc},
where the definitions of feasible sets and free contracts for firms are modified as follows.
We say that $M\subseteq E$ is {\em feasible} for $f\in F$ if $C_f(M_f)=M_f$,
and we say that $e\in M\setminus E$ is {\em free} for $f\coloneqq \f(e)$ if $e\in C_f(M_f+e)$.
Let us call this stability {\em HM stability} to distinguish it from the stability in {\sc smti-olc}.%
\footnote{Hatfield and Milgrom \cite{HM05} defined stability by 
the nonexistence of blocking coalitions rather than blocking pairs. 
Such a definition is identical to ours if the choice functions of firms satisfy substitutability \cite{Goto16, KTY18}.}

Hatfield and Milgrom~\cite{HM05}
showed that the following two conditions for each choice function $C_f:2^{E_f}\to 2^{E_f}$
are essential for strategy-proofness.%
\footnote{To be more precise, Hatfield and Milgrom \cite{HM05} implicitly assumed a condition of 
choice functions called {\em the irrelevance of rejected contracts}. 
Ayg{\"u}n and S{\"o}nmez \cite{AS13} pointed out that this condition is important for the results of \cite{HM05}
and also showed that substitutability and the law of aggregate demand together imply this condition.}
\footnote{In the original model of Hatfield and Milgrom \cite{HM05}, 
it is assumed that a firm's choice function always returns a set 
that does not contain multiple contracts associated with the same worker. However, this assumption is not necessary 
to obtain their results.}
\smallskip
\begin{description}
\item[Substitutability:] $S\subseteq T\subseteq E_f$ implies  $S\setminus C_f(S)\subseteq T\setminus C_f(T)$.
\item[Law of aggregate demand:] $S\subseteq T\subseteq E_f$ implies  $|C_f(S)|\leq |C_f(T)|$.
\end{description}

Hatfield and Milgrom~\cite{HM05}
showed that, if each $C_f$ satisfies substitutability, then 
there exists a unique worker-optimal stable matching.
Furthermore, they provided the following theorem.

\begin{theorem}[Hatfield and Milgrom \cite{HM05}]\label{thm:HM}
In the HM model, if each $C_f$ satisfies substitutability and the law of aggregate demand,
then the mechanism that always returns the worker-optimal HM-stable matching is worker-strategy-proof. 
\end{theorem}

We can reduce {\sc smti-olc} to the HM model if the preference lists of all agents are strict.
Let $I=(W,F,E,\{P_w\}_{w\in W}, \{\laminar_f,q_f,P_f\}_{f\in F})$ be  an {\sc smti-olc} instance without ties.
For each firm $f\in F$, let $(E_f, \I_f)$ be a laminar matroid defined by $\laminar_f$ and $q_f$ 
and let $\succ_f$ be a strict linear order on $E_f$ representing $P_f$.
From an ordered matroid $(E_f, \I_f, \succ_f)$, define $C_f:2^{E_f}\to 2^{E_f}$ as in Section~\ref{sec:matroid-kernel}.
Then, we say that an instance $I'=(W,F,E,\{P_w\}_{w\in W}, \{C_f\}_{f\in F})$ of the HM model is {\em induced from $I$}. 
The following facts are known from previous works.

\begin{proposition}\label{prop:reduction1}
For an {\sc smti-olc} instance $I$ without ties,
the choice functions in the induced instance $I'$ satisfy substitutability and the law of aggregate demand.
\end{proposition}
\begin{proof}
It was shown by Fleiner \cite{Fleiner01, Fleiner03}  that choice functions defined from ordered matroids as in Section~\ref{sec:matroid-kernel} 
satisfy substitutability (called {\em comonotonicity} in \cite{Fleiner03}).
The law of aggregate demand easily follows from the monotonicity of matroid rank functions (see, e.g., \cite{Oxley,Yokoi19}).
\end{proof}

\begin{proposition}\label{prop:reduction2}
For an {\sc smti-olc} instance $I$ without ties,
a set $M\subseteq E$ is an HM-stable matching of the induced instance $I'$
if and only if $M$ is a stable matching of $I$.
\end{proposition}
\begin{proof}
By the definition of $C_f$, we have $C_f(M_f)=M_f$ if and only if $M_f\in \I_f$.
Note that the definition of $C_f$ is identical to the matroid greedy algorithm (see, e.g., Oxley \cite{Oxley}).
Then, if there is a weight function $w:E_f\to \R_+$ such that $w(e)>w(e')\Leftrightarrow e\succ_f e'$,
the set $C_f(M_f+e)$ is the maximum weight independent subset of $M_f+e$ (see also \cite[Proposition~1]{Yokoi17}).
This fact implies that, when $M_f\in \I_f$, we have $e\in C_f(M_f+e)$ if and only if $M_f+e\in \I_f$ or 
there exists  $e'\in M_{f}$ such that $e\succ_f e'$ and $M_{f}+e-e'\in \I_f$.
Then, the statement follows.
\end{proof}

\begin{proof}[Proof of Lemma~\ref{lem:SP}]
By combining Theorem~\ref{thm:HM} and Propositions~\ref{prop:reduction1} and \ref{prop:reduction2}, we can immediately obtain Lemma~\ref{lem:SP}.
\end{proof}

\subsection{Proof of Theorem~\ref{thm:2-approx}}\label{app:2-approx}
We prove Theorem~\ref{thm:2-approx}, which states that 
there is a worker-strategy-proof mechanism returning a $2$-approximate solution for {\sc max-smti-lc}.
For this purpose, we prepare the following proposition,
which generalizes a well-known fact of {\sc max-smti} to the setting of {\sc max-smti-lc}.
It claims that we can obtain a $2$-approximate solution by breaking ties arbitrarily 
and computing a stable matching of the resultant instance \cite{MIIMM02}.
\begin{proposition}\label{prop:2-approx}
For an {\sc smti-lc} instance $I=(W,F,E,\{\laminar_a,q_a,P_a\}_{a\in W\cup F})$, 
define $I'$ by replacing each $P_a$ with any strict preference $P'_a$ that is consistent with $P_a$ 
(i.e., obtain $I'$ from $I$ by tie-breaking). 
Then, any stable matching $M$ of $I'$ is a stable matching of $I$ and satisfies $|M|\geq \frac{1}{2}|M_{\rm OPT}|$, 
where $M_{\rm OPT}$ is a maximum-cardinality stable matching of $I$.
\end{proposition}
\begin{proof}
First, we show that $M$ is a stable matching of $I$.
As $M$ is a matching in $I'$, it is clearly a matching in $I$.
Suppose, to the contrary, that some contract $e\in E\setminus M$ blocks $M$ in $I$.
Then, $e$ is free for $w\coloneqq \w(e)$,
which implies that $M_w+e$ is feasible for $w$ or 
there exists $e'\in M_w$ such that $e\succ_w e'$ and $M_w+e-e'$ is feasible for $w$,
where $\succ_w$ is defined by $P_w$.
As $P'_w$ is consistent with $P_w$, it implies that $e$ is free for $w$ also in $I'$.
We can similarly show that $e$ is free for $f\coloneqq \f(e)$ in $I'$. Then, $e$ blocks $M$ in $I'$, 
which contradicts $M$ being a stable matching of $I'$.

Next, we show that $|M|\geq \frac{1}{2}|N|$, where $N\coloneqq M_{\rm OPT}$.
Let $(E,\I_W)$ and $(E,\I_F)$ be defined as in Section~\ref{sec:correctness} (after Lemma~\ref{lem:stable}).
Then, we have $M,N\in \I_W\cap \I_F$.
Suppose, to the contrary,  that $|M|<\frac{1}{2}|N|$.
As $(E,\I_W)$ is a matroid, the matroid axiom (I3) implies that 
there exists a subset $A_1\subseteq N\setminus M$ such that
$|A_1|=|N|-|M|$ and $M\cup A_1\in \I_W$.
Similarly, as $(E,\I_F)$ is a matroid, there exists a subset
$A_2\subseteq N\setminus M$ such that
$|A_2|=|N|-|M|$ and $M\cup A_2\in \I_F$.
By $|M|<\frac{1}{2}|N|$, we have $|A_1|+|A_2|=2(|N|-|M|)>|N|\geq |N\setminus M|$.
Since $A_1,A_2\subseteq  N\setminus M$, this implies $A_1\cap A_2\neq \emptyset$.
Then, there exists $e\in A_1\cap A_2$, which satisfies $M+e\in \I_W$ and
$M+e\in \I_F$. Then, $e$ is free for both the worker $\w(e)$ and the firm $\f(e)$ in $I'$.
This contradicts the stability of $M$ in $I'$.
\end{proof}

We now complete the proof of Theorem~\ref{thm:2-approx}
by combining Lemma~\ref{lem:SP} and Proposition~\ref{prop:2-approx}.

\begin{proof}[Proof of Theorem~\ref{thm:2-approx}]
Define a mechanism $A$ as follows.
Given an instance $I$ of {\sc smti-olc}, break ties such that, among indifferent contracts, 
contracts with smaller indices have higher priorities. 
Let $I'$ be the resultant instance, and let $A(I)$ be the worker-optimal stable matching of $I'$. 

By Proposition~\ref{prop:2-approx}, the matching $M\coloneqq A(I)$ satisfies $|M|\geq \frac{1}{2}|M_{\rm OPT}|$.
Furthermore, the time complexity follows from Theorem~\ref{thm:Fleiner} and Lemma~\ref{lem:oracle}.
We complete the proof by showing the worker-strategy-proofness of $A$.

Let $J$ be an instance obtained from $I$ by replacing the preference list of some worker $w$ with some other list.
Then, $N\coloneqq A(J)$ is the worker-optimal stable matching of $J'$,
where $J'$ is obtained from $J$ by breaking ties according to the above-mentioned tie-breaking rule.
Then, $I'$ and $J'$ differ only in the preference lists of $w$.
Let $P_w$ and $P'_w$ be the preference lists of $w$ in $I$ and $I'$, respectively.
By Lemma~\ref{lem:SP}, $w$ weakly prefers $M$ to $N$ with respect to $P'_w$.
As $P'_w$ is consistent with $P_w$, 
we can see that $w$ weakly prefers $M$ to $N$ also with respect to $P_w$.
Thus, the mechanism $A$ is worker-strategy-proof.
\end{proof}

\end{document}